\documentclass[letterpaper, 10 pt, conference]{ieeeconf}  % Comment this line out if you need a4paper

\IEEEoverridecommandlockouts                              % This command is only needed if
\makeatletter
\def\endthebibliography{%
	\def\@noitemerr{\@latex@warning{Empty `thebibliography' environment}}%
	\endlist
}
\makeatother

\usepackage{cite}
\usepackage{amsmath,amssymb,amsfonts}

\usepackage{graphicx}
\usepackage{textcomp}
\usepackage{xcolor,soul}

\usepackage{algorithmic}

\usepackage{mathtools}

\usepackage{mathrsfs}

\usepackage{todonotes}
\usepackage{xargs}

\newtheorem{remark}{\textbf{Remark}}
\newtheorem{lem}{\textbf{Lemma}}
\newtheorem{theorem}{\textbf{Theorem}}

\newtheorem{assumption}{\textbf{Assumption}}
\newtheorem{definition}{\textbf{Definition}}
\newtheorem{corollary}{\textbf{Corollary}}

\renewcommand{\baselinestretch}{0.96}

\newcommandx{\nader}[2][1=]{\todo[linecolor=orange,backgroundcolor=orange!25,bordercolor=orange,author=Nader,#1]{#2}}

\newcommandx{\mehdi}[2][1=]{\todo[linecolor=green,backgroundcolor=green!25,bordercolor=green,author=Mehdi,#1]{#2}}

\title{\LARGE \bf
A Dynamic Coding Scheme to Prevent Covert Cyber-Attacks in Cyber-Physical Systems}

\author{Mahdi Taheri$^{1}$, Khashayar Khorasani$^{2}$, and Nader Meskin$^{3}$
\thanks{$^{1}$Mahdi Taheri (mtaheri@caltech.edu) is with the Division of Engineering and Applied Science, California Institute of Technology, Pasadena, CA, USA. This work was conducted during his Ph.D. studies at Concordia University.}%
\thanks{$^{2}$Khashayar Khorasani (kash@ece.concordia.ca) is with the Department of Electrical and Computer Engineering, Concordia University, Montreal, Canada.}%
\thanks{$^{3}$Nader Meskin (nader.meskin@qu.edu.qa) is with the Department of Electrical Engineering, Qatar University, Doha, Qatar.}%
\thanks{
	The authors would like to acknowledge the financial support received from NATO under the Emerging Security Challenges Division program. K. Khorasani and N. Meskin would like to acknowledge the support received from NPRP grant number 10-0105-17017 from the Qatar National Research Fund (a member of Qatar Foundation). K. Khorasani would also like to acknowledge the support received from the Natural Sciences and Engineering Research Council of Canada (NSERC) and the Department of National Defence (DND) under the Discovery Grant and DND Supplemental Programs. This work was also supported in part by funding from the Innovation for Defence Excellence and Security (IDEaS) program from	the Department of National Defence (DND). Any opinions and conclusions in this work are strictly those of the authors and do not reflect the views,
	positions, or policies of - and are not endorsed by - IDEaS, DND, or the Government of Canada.}
}

\begin{document}

\maketitle
\thispagestyle{empty}
\pagestyle{empty}

%%%%%%%%%%%%%%%%%%%%%%%%%%%%%%%%%%%%%%%%%%%%%%%%%%%%%%%%%%%%%%%%%%%%%%%%%%%%%%%%
\begin{abstract}
	In this paper, we address two main problems in the context of covert cyber-attacks in cyber-physical systems (CPS). First, we aim to investigate and develop necessary and sufficient conditions in terms of disruption resources of the CPS that enable adversaries to execute covert cyber-attacks. These conditions can be utilized to identify the input and output communication channels that are needed by adversaries to execute these attacks. Second, this paper introduces and develops a dynamic coding scheme as a countermeasure against covert cyber-attacks. Under certain conditions and assuming the existence of one secure input and two secure output communication channels, the proposed dynamic coding scheme prevents adversaries from executing covert cyber-attacks. A numerical case study of a flight control system is provided to demonstrate the capabilities of our proposed and developed dynamic coding scheme. 
\end{abstract}

%%%%%%%%%%%%%%%%%%%%%%%%%%%%%%%%%%%%%%%%%%%%%%%%%%%%%%%%%%%%%%%%%%%%%%%%%%%%%%%%
\section{Introduction}
Addressing the problem of cyber-threats in cyber-physical systems (CPS) has attracted a significant amount of attention over the past two decades \cite{cardenas2008secure,AAPC,rditp,adaii,ascf,10298792}. Several cases of cyber-attacks in CPS such as in transportation systems, power systems and smart grids, and industrial process control systems have been reported in recent years \cite{sandberg2022secure}. Consequently, it is imperative to study cyber-attacks in CPS and investigate countermeasure strategies for them.

A certain class of cyber-attacks is referred to as \textit{perfectly undetectable} since the impact of the attack signal cannot be observed in the sensor readings of the CPS \cite{9039550}. One of the cyber-attacks that belongs to this class is known as covert attack. In a covert attack, adversaries inject their desired attack signals into the input communication channels of the CPS and eliminate the impact of the injected signals from sensor readings by compromising the output communication channels \cite{smith2011decoupled,adaii}. Consequently, adversaries need to have access to disruption resources and system knowledge to carry out their covert cyber-attacks \cite{ascf}.

In \cite{ascf}, in order to execute covert cyber-attacks, it has been stated that adversaries need to have access to all actuators and sensors and to have a full system knowledge. On the other hand, in \cite{9683376}, authors have shown that in order to perform covert attacks, it is sufficient for adversaries to know the Markov parameters of the CPS and to have access to a certain number of input and output communication channels. Moreover, in \cite{kwon2013security}, necessary and sufficient conditions for performing covert attacks are derived in terms of the span of the controllability matrix of the CPS and its connection with input channels that are accessible by adversaries. However, one requires analyzing the span of the controllability matrix to investigate the stated necessary and sufficient conditions for performing covert attacks, which may be challenging in case of large-scale CPS. Consequently, developing and investigating readily verifiable necessary and sufficient conditions on the required disruption resources of the CPS for executing covert cyber-attacks is one of the problems that we aim to address in this paper.

Considering that the impact of covert cyber-attacks cannot be observed in the sensor readings, it is challenging to detect them. Centralized and distributed observer-based methodologies have been developed in the literature to detect covert cyber-attacks in large-scale interconnected CPS \cite{9104026,10298792}. In addition to observer-based methodologies, active cyber-attack detection methodologies such as watermarking methods, modulations, and coding schemes have been utilized in the literature to detect stealthy cyber-attacks \cite{docaazdaicps,ferrari2020switching,7011170,miao2016coding,fang2019two,9683376,9683075}. In \cite{docaazdaicps}, a modulation matrix has been employed on the input communication channels that disrupts the adversary's perfect system knowledge for performing covert attacks. Moreover, authors in \cite{9683075} have developed and proposed a multiplicative watermarking method that relies on the output-to-output gain between the residual of a detection method and certain performance metrics. Authors in \cite{9683376} have focused on targeting the required disruption resources for adversaries in implementing covert attacks. In \cite{9683376}, a coding scheme is proposed that under certain conditions increases the number of required output communication channels that should be compromised by adversaries to its maximum value. %Under the assumption that adversaries can manipulate only one actuator at any instant of time, by utilizing the proposed coding scheme in \cite{9683376}, all the output communication channels need to be compromised to carry out a covert attack. Hence, securing only one output communication channel will prevent adversaries from executing covert attacks. However, the assumption on manipulating only one input communication channel is restrictive.

In this paper we develop a dynamic coding scheme consisting of an encoder on the C\&C side of the CPS and a decoder on the plant side. The control input from the C\&C side will be encoded using the encoder and eventually it will be decoded on the plant side. Hence, our proposed dynamic coding scheme does not affect the performance and control objectives of the CPS. As the main assumption of this work, we consider the existence of one secure input and two secure output communication channels. Hence, under certain conditions, even if adversaries have knowledge of the parameters of the coding scheme, they will not be able to perform covert cyber-attacks in the CPS. %Moreover, as opposed to \cite{9683376}, in this paper, we do  not impose any conditions on the number of input communication channels that can be manipulated by adversaries and we assume only one secure input channel.

In summary, the contributions of this paper are as follows:

\begin{enumerate}
	\item Necessary and sufficient conditions under which covert cyber-attacks can be performed in the CPS are derived. The developed conditions can be used to determine which disruption resources in terms of input and output communication channels of the CPS should be compromised to carry out covert attacks.
	\item As an active countermeasure against covert attacks, we develop and propose a dynamic coding scheme. The proposed dynamic coding scheme includes an encoder on the C\&C side and a decoder on the plant side of the CPS. Under certain conditions, if there exists one secure input and two secure output communication channels, adversaries will not be capable of performing covert cyber-attacks in the CPS.
\end{enumerate}

As opposed to the derived necessary and sufficient conditions for performing covert attacks in \cite{kwon2013security} that require determining the span of the controllability matrix of the CPS, in our first contribution, we develop necessary and sufficient conditions that rely on the relative degree of the CPS which can be computed using simple matrix multiplications. As for the second contribution of this work, we develop a dynamic coding scheme for preventing adversaries from performing covert attacks that relaxes the assumption of having only one compromised input communication channel for the static coding scheme that is proposed in \cite{9683376}. Moreover, in \cite{ferrari2020switching}, it is assumed that adversaries do not know the parameters of the watermarking scheme. However, in this work, we assume that adversaries know the parameters of our dynamic coding scheme. Consequently, in comparison to \cite{ferrari2020switching} where adversaries' system knowledge is compromised, in our proposed dynamic coding scheme, adversaries' disruption resources are targeted.

\section{Problem Statement and Formulation}\label{s:problem}
\subsection{State-Space Representation of the CPS Under Cyber-Attacks}
Let us consider a discrete-time linear time-invariant cyber-physical system (CPS) represented as follows:
\begin{align}\label{e:CPS}
	x(k+1)=&Ax(k)+Bu(k)+\omega(k), \nonumber \\
	y(k)=& Cx(k)+\nu(k),
\end{align}
where $x(k) \in \mathbb{R}^n$, $u(k)\in\mathbb{R}^m$, and $y(k) \in \mathbb{R}^p$ denote the state, the control input, and the sensor measurement, respectively. The process and measurement noise are denoted by $\omega(k) \in \mathbb{R}^n$ and $\nu(k) \in \mathbb{R}^p$ as zero-mean Gaussian distributions, respectively. Matrices $A$, $B$, and $C$ have appropriate dimensions.

In presence of actuator and sensor cyber-attacks, the CPS \eqref{e:CPS} can be expressed as:
\begin{align}\label{e:CPS_attacked}
	x(k+1)=&Ax(k)+B(u(k)+L_\text{a} a_\text{u}(k))+\omega(k), \nonumber \\
	y(k)=& Cx(k)+D_\text{a}a_\text{y}(k)+\nu(k),
\end{align}
where $a_\text{u}(k) \in \mathbb{R}^{m_\text{a}}$, $m_\text{a}\leq m$ represents the actuator attack signal and $a_\text{y}(k) \in \mathbb{R}^{p_\text{a}}$, $p_\text{a}\leq p$  denotes the sensor attack signal. Additionally, the signatures of actuator and sensor cyber-attack signals are captured by $B_\text{a}=BL_\text{a}$ and $D_\text{a}$, respectively. Furthermore, in our derivation, we omit the consideration of noise effects on the CPS and assume that $\omega(k)=0$ and $\nu(k)=0$ for all $k\geq 0$. Nevertheless, to demonstrate the robustness of our proposed methodologies against noise, the numerical case study in Section~\ref{s:simu} is conducted in presence of both process and sensor noise.

Consider $\mathcal{I}_{\text{a}}=\{u_1^{\text{a}},\dots,u_{m_\text{a}}^{\text{a}}\}$ and $\mathcal{S}_{\text{a}}=\{s_1^{\text{a}},\dots,s_{p_\text{a}}^{\text{a}}\}$ as the index sets of compromised input and output communication channels with $|\mathcal{I}_{\text{a}}|=m_\text{a}$ and $|\mathcal{S}_{\text{a}}|=p_\text{a}$, respectively, where $u_i^{\text{a}}$ belongs to the set of inputs $\{u_1,\, \dots, u_\text{m}\}$, for $i=1,\dots,m_\text{a}$, $s_q^{\text{a}}$ belongs to the set of outputs $\{\,y_1,\, \dots, y_\text{p}\}$, for $q=1,\dots,p_\text{a}$, and $|\cdot|$ denotes the cardinality of a set. The matrices $L_\text{a}$ and $D_\text{a}$ are defined with entries corresponding to the compromised elements in the communication channels. We have $L_\text{a}=[l_{u_1^{\text{a}}},\dots,l_{u_{m_\text{a}}^{\text{a}}}]$, where $l_{i}\in \mathbb{R}^{m}$ is a vector with all its entries equal to zero except for the $i$-th element that is equal to one, for $i=u_1^{\text{a}},\dots,u_{m_\text{a}}^{\text{a}}$. Moreover, $D_\text{a}=[d_{s_1^{\text{a}}},\dots,d_{s_{p_\text{a}}^{\text{a}}}]$, where all entries of $d_{q}\in \mathbb{R}^{p}$ are zero except for the $q$-th element, for $q=s_1^{\text{a}},\dots,s_{p_\text{a}}^{\text{a}}$.

The following definition and lemma are required in the remainder of the paper.
\begin{definition}[Left-Invertibility \cite{ctfls}]\label{def:leftInvertible}
	For system \eqref{e:CPS}, let ${x}(0)=0$, $\omega(k)=0$, and $\nu(k)=0$. The triple $({C},{A},{B})$ is left-invertible if for any ${y}(k)=0$, $\forall k\geq 0$, one has ${u}(k)= 0$, $\forall k\geq 0$.
\end{definition}

\begin{lem}[\cite{ctfls}]\label{lem:leftInv}
	The triple $({C},{A},{B})$ is left-invertible if and only if $\text{rank}(P(\lambda))=n+m$ for all but finitely many $\lambda\in\mathbb{C}$, where 
	\begin{equation}
		P(\lambda)=\begin{bmatrix}
			\lambda I-A & -B \\
			C & 0
		\end{bmatrix},
	\end{equation}
	is the Rosenbrock system matrix.
\end{lem}

\subsection{Input/Output Model of the CPS Under Cyber-Attacks}\label{ss:stealthy}
In this section, we present the Input/Output (I/O) model for the CPS \eqref{e:CPS_attacked} within a specified time window. Subsequently, we define perfectly undetectable cyber-attacks in terms of the I/O model.

The I/O model of the CPS \eqref{e:CPS_attacked} over the time window $\{0,1,\dots,N-1\}$ for $N\geq n$ can be represented as:
\begin{equation}\label{e:I/O}
	\begin{split}
		Y(N)=&\mathcal{O}_\text{N} x(0)+\mathcal{C}_\text{N} U(N)+\mathcal{C}_\text{a}U_\text{a}(N) + \mathcal{D}_\text{a}Y_\text{a}(N),
	\end{split}
\end{equation}
where $Y(N)=[y(0)^\top, y(1)^\top, \dots , y(N-1)^\top]^\top$ represents the output of the I/O model. Other vectors, i.e. $U(N)$ for inputs, $U_\text{a}(N)$ for actuator attack signals, and $Y_\text{a}(N)$ for sensor attack signals, are defined in a similar way. Moreover, $\mathcal{D}_\text{a}=I_N \otimes {D}_\text{a}$.

The matrices $\mathcal{O}_\text{N}$, $\mathcal{C}_\text{N}$, and $\mathcal{C}_\text{a}$ are structured as follows:
\begin{align}\label{e:obs_markov}
	\mathcal{O}_\text{N}=& \begin{bmatrix}
		C \\
		CA \\
		\vdots \\
		CA^{N-1}
	\end{bmatrix}, \, \mathcal{C}_\text{N}=\begin{bmatrix}
		0 & 0 & \cdots & 0 \\
		CB & 0 & \cdots & 0 \\
		\vdots & \vdots & \ddots & \vdots \\
		CA^{N-2}B & CA^{N-3}B & \cdots & 0
	\end{bmatrix}, \\
	\mathcal{C}_\text{a}=& \begin{bmatrix}
		0 & 0 & \cdots & 0 \\
		CB_\text{a} & 0 & \cdots & 0 \\
		\vdots & \vdots & \ddots & \vdots \\
		CA^{N-2}B_\text{a} & CA^{N-3}B_\text{a} & \cdots & 0
	\end{bmatrix}.
\end{align}

Let $\mathcal{Y}(x(0),U(N),\check{U}_\text{a}(N))$ denote the output of the I/O model in \eqref{e:I/O} over the time window $\{0,1,\dots,N-~1\}$. This function depends on the initial state $x(0)$, the vector of control inputs $U(N)$, and the vector of attack signals denoted by $\check{U}_\text{a}(N)=[U_\text{a}(N)^\top , Y_\text{a}(N)^\top]^\top$.

\subsection{Objectives}\label{ss:obj}
We have two objectives in this paper. Our first objective is to investigate and study necessary and sufficient conditions in terms of disruption resources under which adversaries are capable of performing covert cyber-attacks. These conditions determine input and output communication channels that should be attacked by adversaries to execute covert attacks in the CPS. As for our second objective, we develop and propose a dynamic coding scheme as a countermeasure against covert attacks that could be utilized by the CPS operators. Hence, in presence of the proposed dynamic coding scheme, if CPS operators secure $1$ input and $2$ output communication channels, adversaries will not be capable of performing covert cyber-attacks in the CPS.

\section{Covert Cyber-Attacks}
According to \cite{adaii}, in the case of covert attacks, adversaries compromise both input and output communication channels of the CPS and design their attack signals $a_\text{u}(k)\neq 0$ and $a_\text{y}(k)\neq 0$ such that the impact of actuator attacks cannot be observed in the transmitted sensor measurements to the C\&C side of the CPS. In the following, a definition for covert cyber-attacks is given in terms of the I/O model of the CPS \eqref{e:I/O}.
\begin{definition}[Covert Cyber-Attacks]\label{def:covert_IO}
	The attack signal $\check{U}_\text{a}(N)=[U_\text{a}(N)^\top , Y_\text{a}(N)^\top]^\top$ in the I/O model of the CPS \eqref{e:I/O} is designated as a covert cyber-attack if for $U_\text{a}(N)\neq0$ and $Y_\text{a}(N)\neq0$, one has $\mathcal{Y}(x(0),U(N),\check{U}_\text{a}(N))=\mathcal{Y}(x(0),U(N),0)$, $\forall N \geq 1$ for any $x(0)$ and $U(N)$, i.e., $\mathcal{C}_\text{a}U_\text{a}(N) + \mathcal{D}_\text{a}Y_\text{a}(N)= 0$.
\end{definition}

Since in Definition~\ref{def:covert_IO}, it is considered that $a_\text{u}(k)$ and $a_\text{y}(k)$ are nonzero, in the case of covert attacks, there should exist at least one compromised actuator and one compromised sensor in the CPS. In order to further investigate covert cyber-attacks, we need to define a relative degree of the CPS \eqref{e:CPS_attacked}.

\begin{definition}\label{def:relative_deg}
	Consider the CPS \eqref{e:CPS_attacked}. Let $C_q$ denote the $q$-th row of $C$, for $q = 1, \dots, p$. The relative degree of the $q$-th output with respect to $a_{\text{u}}(k)$ is $r_{\text{a}}^q$ if $C_q A^{i_q-1}B_{\text{a}} = 0$, $\forall i_q = 1, \dots, r_{\text{a}}^q-1$, and $C_q A^{r_{\text{a}}^q-1}B_{\text{a}} \neq 0$. If ~$C_q A^{i_q-1}B_{\text{a}} =~0$ for any positive integer $i_q$, the relative degree of the $q$-th output with respect to $a_{\text{u}}(k)$ is undefined and we set $r_{\text{a}}^q=\infty$. Moreover, one has $r_\text{a}=\text{min}\{r_\text{a}^1,\dots,r_\text{a}^p\}$.
\end{definition}

The relative degree $r_\text{a}^q$ quantifies the delay before the actuator attack signal $a_\text{u}(k)$ affects the $q$-th output. If the relative degree is undefined, i.e., $r_{\text{a}}^q=\infty$, it implies that the attack signal does not influence the $q$-th output at any finite time, which implies that the attack is entirely decoupled from that particular output.

We are now in a position to provide necessary and sufficient conditions under which covert cyber-attacks in the sense of Definition~\ref{def:covert_IO} can be carried out in the I/O model of the CPS \eqref{e:I/O}. 

\begin{theorem}\label{th:covert}
	Given \underline{any actuator cyber-attack signal} $a_\text{u}(k)\in \mathbb{R}^{m_\text{a}}$, a covert attack in the sense of Definition~\ref{def:covert_IO} can be executed in the CPS \eqref{e:I/O} if and only if relative degrees for corresponding outputs of triples $(C,A,B_\text{a})$ and $(D_\text{a}^* C,A,B_\text{a})$ are either equal or cannot be defined as per Definition~\ref{def:relative_deg}, where $D_\text{a}^*=D_\text{a}D_\text{a}^\top$, i.e., $C_{q}A^{r_\text{a}^{q}-1}B_\text{a}=(D_\text{a}^*C)_{q}A^{r_\text{a}^{q}-1}B_\text{a}$, for $q=1,\dots,p$, where $(D_\text{a}^*C)_{q}$ denotes the $q$-th row of $D_\text{a}^*C$. 
\end{theorem}

\begin{proof}
	It follows from Definition~\ref{def:covert_IO} that in the case of covert attacks, adversaries should design the actuator attack signal $a_\text{u}(k)$ and the sensor attack signal $a_\text{y}(k)$ such that the following conditions hold:
	\begin{subequations}\label{e:th_covert:c}
		\begin{align}
			& CA^{r_\text{a}-1}B_\text{a}a_\text{u}(0)+D_\text{a}a_\text{y}(r_\text{a})=0, \label{e:th_covert:c1} \\
			& CA^{r_\text{a}}B_\text{a}a_\text{u}(0)+CA^{r_\text{a}-1}B_\text{a}a_\text{u}(1)+D_\text{a}a_\text{y}(r_\text{a}+1)=0, \label{e:th_covert:c2} \\
			& \vdots \nonumber \\
			& CA^{N-2}B_\text{a}a_\text{u}(0)+CA^{N-3}B_\text{a}a_\text{u}(1)+CA^{N-4}B_\text{a}a_\text{u}(2) \nonumber \\
			&+\cdots+CA^{r_\text{a}-1}B_\text{a}a_\text{u}(N-2)+D_\text{a}a_\text{y}(N-1)=0. \label{e:th_covert:cN}
		\end{align}
	\end{subequations}
	
	Consequently, in order to cancel out the impact of actuator attacks in \eqref{e:th_covert:c}, adversaries need to design the sensor attack signal in the following form:
	\begin{equation}\label{e:th_covert:ay}
		a_\text{y}(j)=-\sum_{\gamma=0}^{j-r_\text{a}}D_\text{a}^\top CA^{r_\text{a}+\gamma-1}B_\text{a}a_\text{u}(j-\gamma-r_\text{a}),
	\end{equation}
	for $j=r_\text{a}, \dots, N-1$, and $a_\text{y}(j)=0$ for $j=0,\dots,r_\text{a}-1$.

	\textbf{Necessary Condition:} 
	Suppose \eqref{e:th_covert:c} holds and relative degrees of triples $(C,A,B_\text{a})$ and $(D_\text{a}^* C,A,B_\text{a})$ are not equal. Given the definition of $D_\text{a}^*$, the relative degree of the $q$-th row of $D_\text{a}^*C$ with respect to the actuator attack signal $a_\text{u}(k)$ is either equal to that of $C_q$ and $a_\text{u}(k)$ or does not exist. Consider $\hat{s} \in \mathcal{S}_\text{a}$ as a sensor at which $C_{\hat{s}}A^{r_\text{a}^{\hat{s}}-1}B_\text{a} \neq 0$ and $(D_\text{a}^*C)_{\hat{s}}A^{r_\text{a}^{\hat{s}}-1}B_\text{a}=0$, where $(D_\text{a}^*C)_{\hat{s}}$ denotes the $\hat{s}$-th row of $D_\text{a}^*C$. 
	
	Considering \eqref{e:th_covert:ay}, let us rewrite the left-hand side of the $j$-th equation of \eqref{e:th_covert:c} in the following form:
	\begin{equation}\label{e:th_covert:c_j}
		[(I_p -D_\text{a}^*)CA^{r_\text{a}-1}B_\text{a} \cdots (I_p -D_\text{a}^*)CA^{r_\text{a}+j-1}B_\text{a}]\begin{bmatrix}
			a_\text{u}(j-1) \\
			\vdots \\
			a_\text{u}(0)
		\end{bmatrix},
	\end{equation}
	for $j=r_\text{a},\dots, N-1$. Consequently, \eqref{e:th_covert:c} holds if and only if all rows of \eqref{e:th_covert:c_j} are equal to zero. However, since we are considering any actuator attack signal, and $C_{\hat{s}}A^{r_\text{a}^{\hat{s}}-1}B_\text{a} \neq 0$ and $(D_\text{a}^*C)_{\hat{s}}A^{r_{\hat{s}}-1}B_\text{a}=0$, there exists at least one nonzero row in \eqref{e:th_covert:c_j}, i.e., the row that corresponds to the $\hat{s}$-th sensor, which contradicts the assumption.
	
	\textbf{Sufficient Condition:} Assume that the relative degrees of triples $(C,A,B_\text{a})$ and $(D_\text{a}^* C,A,B_\text{a})$ for all the sensors are equal. Hence, $C_qA^iB_\text{a}=(D_\text{a}^*C)_qA^iB_\text{a}$, for any $i \in \mathbb{N}$ and $\forall q \in \{1,\dots,p\}$. The latter implies that \eqref{e:th_covert:ay} is the solution to \eqref{e:th_covert:c}, and all rows in \eqref{e:th_covert:c_j} are equal to zero. This completes the proof of the theorem. 
\end{proof}

\begin{remark}\label{rem:upper_bound}
	Consider the case where the hypothesis of Theorem~\ref{th:covert} does not hold. Consequently, as per Definition~\ref{def:leftInvertible}, if the triple $((I_p-D_\text{a}^*)C,A,B_\text{a})$ \underline{is not left-invertible} (see \cite{ctfls} for more details), adversaries can \underline{design} a certain actuator attack signal $a_\text{u}(k)$ that makes \eqref{e:th_covert:c_j} equal to zero despite having $(I_p -D_\text{a}^*)CA^{r_\text{a}-1}B_\text{a}\neq~0$. The latter implies that adversaries can design an actuator attack signal such that its impact on the not compromised sensors, i.e, $I_p -D_\text{a}^*$, becomes zero, while the impact of an arbitrary actuator attack signal would show up in the not compromised sensors. However, in this case, the actuator attack signal $a_\text{u}(k)$ should be specifically designed to make \eqref{e:th_covert:c_j} equal to zero whereas in Theorem~\ref{th:covert}, if the conditions are satisfied, adversaries can perform covert attacks for any arbitrary actuator attack signal. Hence, it can be concluded that Theorem~\ref{th:covert} provides necessary and sufficient conditions for covert attacks when the actuator attack signal is arbitrary.
\end{remark}

\textbf{Example 1:} Let $A=B=C=I_2$, where $I_2=[1, \, 0;0, \; 1]$. Moreover, consider the case where both actuators and the first sensor are compromised, i.e., $B_\text{a}=B$ and $D_\text{a}=[1, \, 0]^\top$. Relative degrees for the first and the second outputs of the triple $(C,A,B_\text{a})$ are equal to $1$. Also, the relative degree corresponding to the first row of $D_\text{a}^*C$ for the triple $(D_\text{a}^*C,A,B_\text{a})$ is $1$ and that for the second row is undefined, i.e, $r_\text{a}^2=\infty$. Hence, the hypothesis of Theorem~\ref{th:covert} does not hold since $C_{q}A^{r_\text{a}^{q}-1}B_\text{a}\neq(D_\text{a}^*C)_{q}A^{r_\text{a}^{q}-1}B_\text{a}$ for $q=2$, which implies that a covert cyber-attack with an arbitrary actuator attack signal cannot be executed. However, following Remark~\ref{rem:upper_bound}, the triple $((I_p-D_\text{a}^*)C,A,B_\text{a})$ is not left-invertible according to Lemma~\ref{lem:leftInv}. Therefore, if one designs an actuator attack signal $a_u(k)=[1, \, 0]^\top$, it will only impact the first sensor, since $C_{2}A^{k}B_\text{a}=[0, \, 1]$, and one has $C_{2}A^{k}B_\text{a}a_u(k)=0$, $\forall k\geq 0$. Considering that the first sensor is compromised, adversaries can design a sensor attack signal $a_\text{y}(k)$ to remove the impact of $a_\text{u}(k)$ from the sensor measurement and perform a covert cyber-attack in the sense of Definition~\ref{def:covert_IO}.

\begin{corollary}\label{cor:covert}
	Assume that the hypothesis of Theorem~\ref{th:covert} holds. Given any $a_\text{u}(k)\in\mathbb{R}^{m_\text{a}}$, the sensor attack signal in a covert cyber-attack can be expressed by
	\begin{equation*}
		a_\text{y}(k)=-\sum_{\gamma=0}^{k-r_\text{a}}D_\text{a}^\top CA^{r_\text{a}+\gamma-1}B_\text{a}a_\text{u}(k-\gamma-r_\text{a}),
	\end{equation*}
	for $k\geq r_\text{a}$, and $a_\text{y}(k)=0$ for $k<r_\text{a}$.
\end{corollary}
\begin{proof}
	The proof follows along	similar lines to that of Theorem~\ref{th:covert} and is omitted for the sake of brevity.
\end{proof}

The conditions derived in Theorem~\ref{th:covert} and Corollary~\ref{cor:covert} for the CPS \eqref{e:CPS_attacked} rely on certain elements within $\mathcal{C}_\text{N}$, namely, the Markov parameters, which are defined by \eqref{e:obs_markov}. Therefore, through the application of methodologies in \cite{hajdasinski1979realization,de1988geometrical,ljung1999system,dong2011identification}, and by utilizing results presented in Theorem~\ref{th:covert} and Corollary~\ref{cor:covert}, one can investigate the vulnerability of the CPS to covert cyber-attacks by employing a data-driven approach.

\section{The Coding Scheme}\label{s:coding}
In this section, a dynamic coding scheme on the input communication channels is developed and proposed which can be used to prevent adversaries from performing covert attacks. The coding scheme is designed such that having only one secure input and two secure output communication channels will result in preventing adversaries from executing covert cyber-attacks.

\begin{remark}
	It should be noted that the dynamic coding scheme can also be implemented on the output	communication channels. However, the design steps and required assumptions will be similar to the case of having the coding scheme on the input communication channels since the location of the coding scheme does not provide one with extra mathematical redundancy. Also, having coding schemes on both input and output communication channels would provide one with a more complicated design problem which may not necessarily lead to having more relaxed assumptions.
\end{remark}

\begin{figure}[!t]
	\centering
	\centerline{\includegraphics[width=1.2\columnwidth]{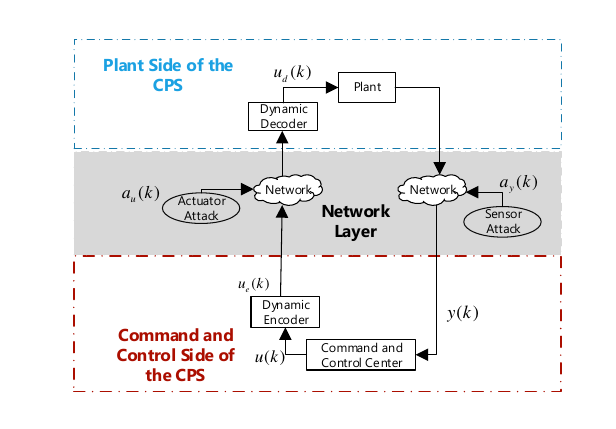}}
	\caption{The dynamic coding scheme, where $u_\text{e}(k)$ is the output of the encoder and $u_\text{d}(k)$ is the output of the decoder.}\label{fig:cps_coding}
\end{figure}

On the Command and Control (C\&C) side of the CPS, an encoder, represented as $\mathcal{E}$, and on the plant side, a decoder, denoted as $\mathcal{D}$, have been designed. The CPS and the dynamic coding scheme are depicted in Fig.~\ref{fig:cps_coding}. Dynamics of the encoder and the decoder which are located on the input communication channels of the CPS, are expressed by
\begin{align}
	\mathcal{E}: & \begin{dcases} \label{e:en}
		x_\text{e}(k+1) &= A_\text{e}x_\text{e}(k)+B_\text{e}u(k), \\
		u_\text{e}(k) &= C_\text{e}x_\text{e}(k)+D_\text{e}u(k),
	\end{dcases} \\
	\mathcal{D}: & \begin{dcases} \label{e:de}
		x_\text{d}(k+1) &= A_\text{d}x_\text{d}(k)+B_\text{d}(u_\text{e}(k)+L_\text{a}a_\text{u}(k)), \\
		u_\text{d}(k) &= C_\text{d}x_\text{d}(k)+D_\text{d}(u_\text{e}(k)+L_\text{a}a_\text{u}(k)),
	\end{dcases}
\end{align}
where $x_\text{e}(k), x_\text{d}(k) \in \mathbb{R}^m$ and $u_\text{e}(k), u_\text{d}(k) \in \mathbb{R}^m$ represent the states and outputs of the encoder $\mathcal{E}$ and the decoder $\mathcal{D}$, respectively. It is important to note that the initial conditions are $x_\text{e}(0)=x_\text{d}(0)=0$. The subsequent lemma lays down the necessary and sufficient conditions that establish the decoder $\mathcal{D}$ as the inverse of $\mathcal{E}$, ensuring that when $a_\text{u}(k)=0$, we have $u_\text{d}(k)=u(k)$ for all $k \geq 0$.

\begin{lem}[\cite{ferrari2020switching}]\label{lem:en_de}
	Let $a_\text{u}(k)=0$. One has $u_\text{d}(k)=u(k)$, $\forall k \geq 0$, if and only if there exists an invertible matrix $T$ satisfying the following conditions:
	\begin{align*}
		D_\text{d}C_\text{e}+C_\text{d}T=0, \, T^{-1}B_\text{d}D_\text{e}=B_\text{e}, \, D_\text{d}=D_\text{e}^{-1}, \\
		T^{-1}A_\text{d}T+T^{-1}B_\text{d}C_\text{e}=T^{-1}A_\text{d}T-B_\text{e}C_\text{d}T=A_\text{e}.
	\end{align*}
\end{lem}

The following assumption holds throughout the paper.

\begin{assumption}\label{assum:en_de}
	The encoder \eqref{e:en} and the decoder \eqref{e:de} are constructed based on Lemma~\ref{lem:en_de}. Additionally, it is assumed that adversaries have knowledge on the parameters of both the encoder $\mathcal{E}$ and the decoder $\mathcal{D}$.
\end{assumption}

In presence of $\mathcal{E}$ and $\mathcal{D}$, the dynamics of the CPS \eqref{e:CPS_attacked} can be described as follows:
\begin{equation}\label{e:CPS_coding}
	\begin{split}
		x(k+1)=&Ax(k)+Bu_\text{d}(k)+\omega(k), \\
		y(k)=& Cx(k)+D_\text{a}a_\text{y}(k)+\nu(k).
	\end{split}
\end{equation}

Consequently, the I/O model of the CPS \eqref{e:CPS_coding} takes the following form:
\begin{equation}\label{e:I/O_coding}
	\begin{split}
		Y(N)=&\mathcal{O}_\text{N} x(0)+\mathcal{C}_\text{N} U(N)+\mathcal{C}_\text{N} \mathcal{C}_\text{d}U_\text{a}(N) + \mathcal{D}_\text{a}Y_\text{a}(N) \\
		&\mathcal{C}_\omega W(N)+V(N),
	\end{split}
\end{equation}
where
\begin{equation*}
	\mathcal{C}_\text{d}= \begin{bmatrix}
		D_\text{d}L_\text{a} & 0 & \cdots & 0 \\
		C_\text{d}B_\text{d}L_\text{a} & D_\text{d}L_\text{a} & \cdots & 0 \\
		\vdots & \vdots & \ddots & \vdots \\
		C_\text{d}A^{N-2}_\text{d}B_\text{d}L_\text{a} & C_\text{d}A^{N-3}_\text{d}B_\text{d}L_\text{a} & \cdots  & D_\text{d}L_\text{a}
	\end{bmatrix}.
\end{equation*}

The presence of the coding scheme leads to the appearance of actuator attack signal impact as $\mathcal{C}_\text{N} \mathcal{C}_\text{d}U_\text{a}(N)$ in the sensor measurements. Given the existence of the coding scheme in \eqref{e:CPS_coding}, we need to redefine the relative degree of our system.

\begin{definition}\label{def:relative_deg_coding}
	The relative degree of the $q$-th output of the CPS \eqref{e:CPS_coding} with respect to the actuator attack signal $a_\text{u}(k)$ is $r_\text{d}^q$ if $C_q A^{i_q}BD_\text{d}L_\text{a} = 0$ for all $i_q<r_\text{d}^q -1$ and $C_q A^{r_\text{d}^q-1}BD_\text{d}L_\text{a} \neq 0$, for $q=1,\dots,p$, where $C_q$ represents the $q$-th row of $C$. If for any positive integer $i_q$ one has $C_q A^{i_q}BD_\text{d}L_\text{a} = 0$, the relative degree for the $q$-th output with respect to $a_\text{u}(k)$ cannot be defined and we set it to $\infty$.  Also, we define $r_\text{d}=\text{min}\{r_\text{d}^1,\dots,r_\text{d}^p\}$.
\end{definition}

The following assumption holds throughout this section.
\begin{assumption}\label{assum:secureChannels}
	For the CPS \eqref{e:CPS_coding}, there exist at least one secure input communication channel and two secure output communication channels, i.e., $\text{rank}(L_\text{a})< m$ and $\text{rank}(D_\text{a})<p-1$.
\end{assumption}

\subsection{Design Specifications of the Dynamic Coding Scheme for Securing the CPS Against Covert Cyber-Attacks}
In presence of the coding scheme and in case of covert attacks, adversaries may try to use sensor attack signals to eliminate the impact of actuator attack signals from output measurements. Hence, the proposed coding scheme should satisfy two requirements. First, when there exists a secure input communication channel, the impact of the coding scheme cannot be canceled out through actuator attack signals. Second, the coding matrices should be designed such that actuator attack signals will affect all the sensor measurements. Under the latter condition, adversaries need to have access to all output communication channels to eliminate the impact of their actuator cyber-attacks from sensor measurements. Design specifications on $\mathcal{E}$ and $\mathcal{D}$ are provided in the following theorem.

\begin{theorem}\label{th:coding_cov}
	Under Assumption~\ref{assum:secureChannels}, consider $q_\text{s1}$ and $q_\text{s2}$ as two secure output communication channels. For the CPS \eqref{e:CPS_coding}, adversaries cannot perform covert cyber-attacks in the sense of Definition~\ref{def:covert_IO} if the following conditions are satisfied: 
	\begin{enumerate}
		\item $C_{q_\text{s1}}A^{r_\text{d}-1}BD_\text{d}L_\text{a}=0$,
		\item $\text{ker}(C_{q_\text{s2}}A^{r_\text{d}-1}BD_\text{d}L_\text{a})\cap\text{ker}(C_{q_\text{s1}}A^{r_\text{d}}BD_\text{d}L_\text{a})=0$,
		\item $C_{q_\text{s1}}A^{r_\text{d}-1}BC_\text{d}B_\text{d}L_\text{a}=0$.
	\end{enumerate}
\end{theorem}
\begin{proof}
	Adversaries need to design their attack signals such that $\mathcal{C}_\text{N} \mathcal{C}_\text{d}U_\text{a}(N) + \mathcal{D}_\text{a}Y_\text{a}(N)=0$. Hence, at the $r_\text{d}$-th and $r_\text{d}+1$-th instances of the output, actuator and sensor attack signals should satisfy the following:
	\begin{subequations}\label{e:th_coding:c}
		\begin{align}
			& CA^{r_\text{d}-1}BD_\text{d}L_\text{a}a_\text{u}(0)+D_\text{a}a_\text{y}(r_\text{d})=0, \label{e:th_coding:c1} \\
			&CA^{r_\text{d}}BD_\text{d}L_\text{a}a_\text{u}(0)+CA^{r_\text{d}-1}B C_\text{d}B_\text{d}L_\text{a}a_\text{u}(0) \nonumber \\
			&+CA^{r_\text{d}-1}BD_\text{d}L_\text{a}a_\text{u}(1)+D_\text{a}a_\text{y}(r_\text{d}+1)=0. \label{e:th_coding:c2}  
		\end{align}
	\end{subequations}
	
	Under Assumption~\ref{assum:secureChannels}, measurements that are transmitted through the communication channels $q_\text{s1}$ and $q_\text{s2}$ cannot be manipulated by means of sensor attacks. Moreover since in Condition 1) we have $C_{q_\text{s1}}A^{r_\text{d}-1}BD_\text{d}L_\text{a}=0$, in order to satisfy \eqref{e:th_coding:c1}, the actuator attack signal should be designed such that $a_\text{u}(0) \in \text{ker}(C_{q_\text{s2}}A^{r_\text{d}-1}BD_\text{d}L_\text{a})$. Furthermore, it follows from Condition 2) that $C_{q_\text{s1}}A^{r_\text{d}}BD_\text{d}L_\text{a}a_\text{u}(0)\neq0$. Since $q_\text{s1}$ is a secure output communication channel and as per Conditions 1) and 3) we have $C_{q_\text{s1}}A^{r_\text{d}-1}BD_\text{d}L_\text{a}a_\text{u}(1)=0$ and $C_{q_\text{s1}}A^{r_\text{d}-1}BC_\text{d}B_\text{d}L_\text{a}a_\text{u}(0)=0$, respectively, the impact of $a_\text{u}(0)$ cannot be removed from the $q_\text{s1}$-th communication channel and \eqref{e:th_coding:c2} cannot be satisfied. This completes the proof of the theorem. 
\end{proof}

The main objective in Theorem~\ref{th:coding_cov} is to design the coding scheme such that the impact of actuator attack signals show up in the sensor measurements that are secured. Moreover, the proposed design specifications in Theorem~\ref{th:coding_cov} ensure that adversaries are not capable of removing the impact of their actuator attacks from sensor measurements that are transmitted through secure output communication channels $q_\text{s1}$ and $q_\text{s2}$. Hence, as the main implication of Theorem~\ref{th:coding_cov}, the CPS operators can prevent adversaries from executing covert attacks by securing $3$ input and output communication channels and employing the proposed coding scheme in this subsection.

According to Theorem~\ref{th:coding_cov}, in order to design the coding scheme, $D_\text{d}$ should satisfy $C_{q_\text{s1}}A^{r_\text{d}-1}BD_\text{d}L_\text{a}=0$ and $\text{ker}(C_{q_\text{s2}}A^{r_\text{d}-1}BD_\text{d}L_\text{a})\cap\text{ker}(C_{q_\text{s1}}A^{r_\text{d}}BD_\text{d}L_\text{a})=0$, simultaneously. Consequently, $C_\text{d}$ and $B_\text{d}$ should satisfy $C_{q_\text{s1}}A^{r_\text{d}-1}BC_\text{d}B_\text{d}L_\text{a}=0$, where one simple design choice could be $C_\text{d}B_\text{d}=D_\text{d}$. Also, Theorem~\ref{th:coding_cov} does not impose any design conditions on $A_\text{d}$. After designing the decoder $\mathcal{D}$, the encoder $\mathcal{E}$ should be developed as per Lemma~\ref{lem:en_de}.

\section{Numerical Case Study: Flight Control System of a Fighter Aircraft}\label{s:simu}
In this section, we study covert cyber-attacks, as per Definition~\ref{def:covert_IO}, in the flight control system of a fighter aircraft. Dynamics of the linearized aircraft system with the sampling period of $T_\text{s}=0.5$ (s) are given by \cite{harkegaard2005resolving,boussaid2014ftc} 
\begin{align*}
	A=&\begin{bmatrix}
		1.0214 & 0.0054 & 0.0003 & 0.4176 & -0.0013 \\
		0 & 0.6307 & 0.0821 & 0 & -0.3792 \\
		0 & -3.4485 & 0.3779 & 0 & 1.1569 \\
		1.1199 & 0.0024 & 0.0001 & 1.0374 & -0.0003 \\
		0 & 0.3802 & -0.0156 & 0 & 0.8062
	\end{bmatrix}, \\
	B=& \begin{bmatrix}
		0.1823 & -0.1798 & -0.1795 & 0.0008 \\
		0 & -0.0639 & 0.0639 & 0.1397 \\
		0 & -1.5840 & 1.5840 & 0.2936 \\
		0.8075 & -0.6456 & -0.6456 & 0.0013 \\
		0 & -0.1005 & 0.1005 & -0.4114
	\end{bmatrix},\\
	C=& \begin{bmatrix}
		1 & 0 & 0 & 0 & 0 \\
		0 & 1 & 0 & 0 & 0 \\
		0 & 0 & 1 & 0 & 0 
	\end{bmatrix}.
\end{align*}

\textbf{Executing Covert Cyber-Attacks (Theorem~\ref{th:covert}):} We consider scenarios where each actuator of the system is attacked separately. If the first input channel is compromised by adversaries, i.e., $L_\text{a}=[1,0,0,0]$, we have $r_\text{a}^1=1$ and $C_1 A^{r_\text{a}^1-1}B_\text{a} = [0.1823,0,0]^\top$, where by definition $B_\text{a}=BL_\text{a}$. Moreover, according to Definition~\ref{def:relative_deg}, assuming that only the first input communication channel is compromised, a relative degree cannot be defined for the second and the third outputs since $C_q A^{i_q}B_\text{a} =0$, for any positive integer $i_q$ and $q=2$ and $3$. Hence, adversaries need to only compromise the first output communication channel to satisfy the condition $C_{1}A^{r_\text{a}^{1}-1}B_\text{a}=(D_\text{a}^*C)_{1}A^{r_\text{a}^{q}-1}B_\text{a}$ in Theorem~\ref{th:covert} and perform a covert cyber-attack, i.e., $D_\text{a}=[1,0,0]^\top$. 

Following Corollary~\ref{cor:covert}, we design a covert attack signal for the case where only the first input and the first output communication channels of the flight control system are under cyber-attacks. As shown in Fig.~\ref{fig:cov_act1}, in presence of system and process noise, a covert cyber-attack is executed and the sensor measurements are close to zero while the values of states are increasing and become unstable.

\begin{figure}[!t]
	\centering
	\centerline{\includegraphics[width=\columnwidth]{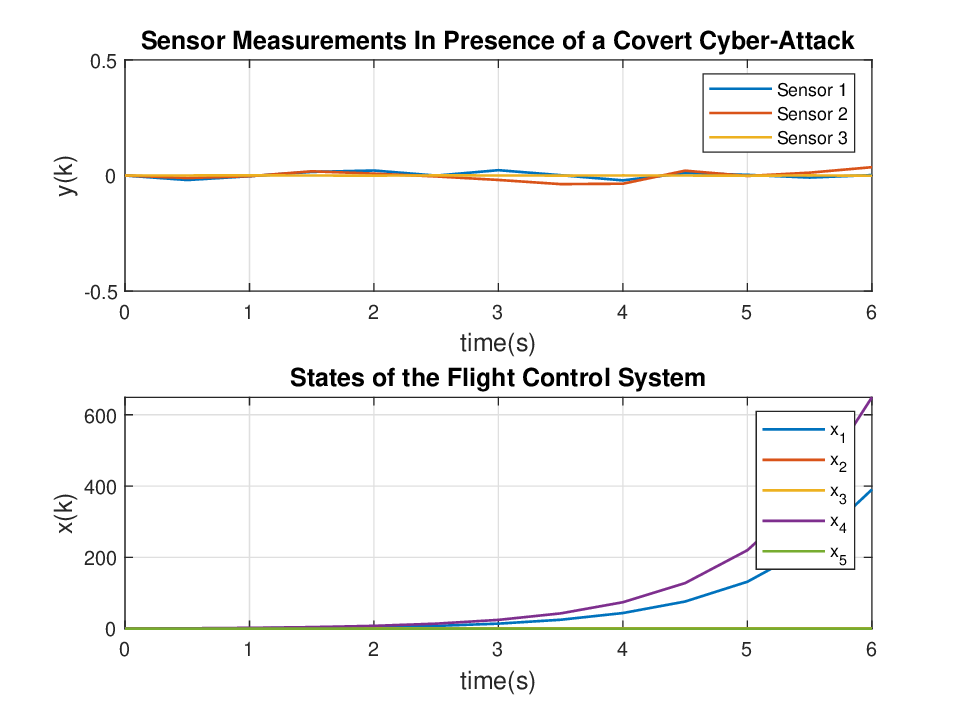}}
	\caption{Covert cyber-attack while the first input and the first output communication channels are compromised.}\label{fig:cov_act1}
\end{figure}

Having either the second or the third or the fourth input communication channels compromised yields a relative degree equal to $1$ for all the outputs, i.e., $r_\text{a}=1$ for $q=1,2,3$ if $L_\text{a}=[0,1,0,0]$ or $L_\text{a}=[0,0,1,0]$ or $L_\text{a}=[0,0,0,1]$. Consequently, in order to execute covert attacks, when any actuator other than the first one is under cyber-attacks, adversaries need to compromise all $3$ output communication channels of the system to satisfy the condition $C_{q}A^{r_\text{a}^{q}-1}B_\text{a}=(D_\text{a}^*C)_{q}A^{r_\text{a}^{q}-1}B_\text{a}$ in Theorem~\ref{th:covert}, for $q=1,\dots,3$.

\textbf{The Coding Scheme (Theorem~\ref{th:coding_cov}):} Our objective is to make the flight control system secure against covert cyber-attacks. Under Assumption~\ref{assum:en_de}, an encoder $\mathcal{E}$ and a decoder $\mathcal{D}$ with their dynamics given by \eqref{e:en} and \eqref{e:de} are considered, respectively. As per Assumption~\ref{assum:secureChannels}, we consider that the actuator $1$ and sensors $1$ and $2$ are secured, i.e., $q_{s1}=1$ and $q_{s2}=2$. Moreover, the decoder $\mathcal{D}$ and the encoder $\mathcal{E}$ satisfy the conditions in Theorem~\ref{th:coding_cov}. The characteristic matrices of the encoder and the decoder are $A_\text{d}=I_4$, $C_\text{d}=-D_\text{d}$, $B_\text{d}=D_\text{d}$, $D_\text{e}=D_{\text{d}}^{-1}$, $B_\text{e}=I_4$, $C_\text{e}=-B_{\text{d}}^{-1}B_\text{e}C_\text{d}$, $A_\text{e}=A_\text{d}-B_\text{e} C_\text{d}$, and
\begin{equation*}
	D_\text{d}=\begin{bmatrix}
		1 & 0.575 & -0.0026 & 0.574 \\
		0 & 0.7911 & 0.0009 & -0.2085 \\
		0 & -0.2085 & 0.0009 & 0.7918 \\
		0 & 0.0009 & 1 & 0.0009
	\end{bmatrix}.
\end{equation*}

Considering that the first and the second sensors are secured, we design the actuator attack signal such that its impact on these two sensors is zero. In order to achieve this goal, we have used the results in \cite{safcpsuua,siolcps} to design the actuator attack signal to be a controllable cyber-attack that its impact cannot be observed in sensors $1$ and $2$. Also, since sensor $3$ is not secured and can be manipulated by adversaries, the sensor attack signal is designed to cancel out the impact of the actuator attack signal on this sensor. Consequently, in Fig.~\ref{fig:cov_3act}, it can be seen that the impact of the actuator attack signal cannot be observed in the sensor readings. As depicted in Fig.~\ref{fig:coding}, in presence of the proposed dynamic coding scheme, adversaries cannot eliminate the impact of their actuator attack signals from sensor measurements despite knowing the parameters of the encoder and the decoder and the attack signal can now be observed in the sensor measurements.

\begin{figure}[!t]
	\centering
	\centerline{\includegraphics[width=\columnwidth]{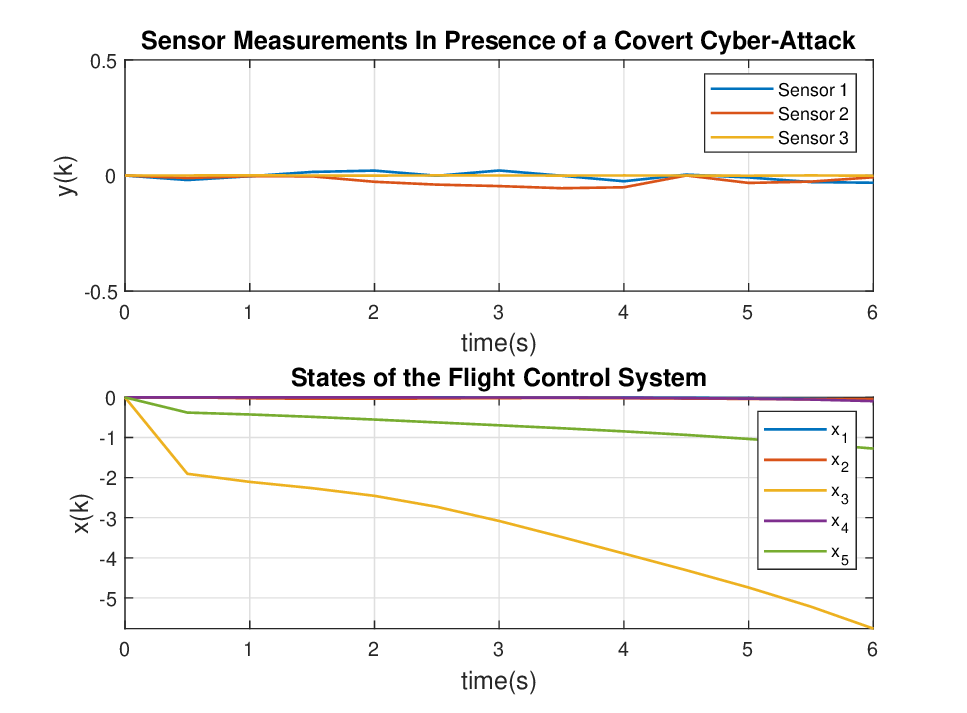}}
	\caption{Covert cyber-attack while actuators $2,\, 3$, and $4$ along with sensor $1$ are compromised.}\label{fig:cov_3act}
\end{figure}

\begin{figure}[!t]
	\centering
	\centerline{\includegraphics[width=\columnwidth]{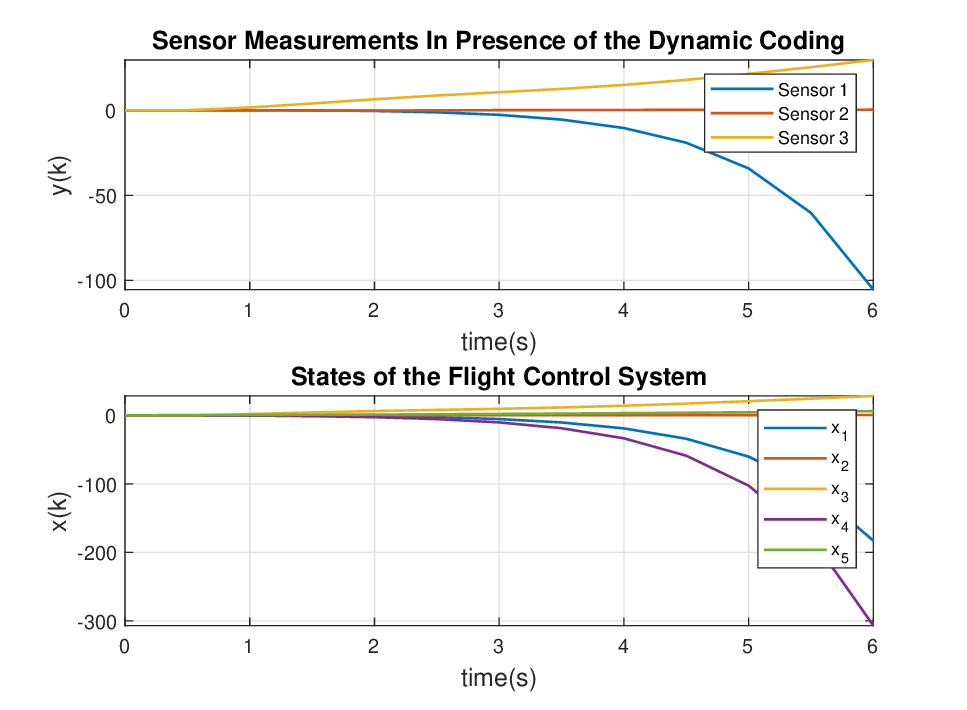}}
	\caption{Covert cyber-attack in presence of the dynamic coding scheme.}\label{fig:coding}
\end{figure}

\section{Conclusion}\label{s:conclu}
This paper has addressed two key challenges related to covert cyber-attacks in cyber-physical systems (CPS). We have investigated and formulated necessary and sufficient conditions in the sense of disruption resources of the CPS that adversaries need in order to carry out covert cyber-attacks. These conditions can be employed to determine the input and output communication channels required for executing covert attacks. As a countermeasure against covert cyber-attacks, a dynamic coding scheme has been introduced and developed. Under certain conditions and assuming the existence of one secure input and two secure output communication channels, the proposed coding scheme effectively prevents adversaries from executing covert cyber-attacks. As for our future work, we will work on developing distributed dynamic coding schemes to prevent undetectable cyber-attacks in multi-agent systems.

\bibliographystyle{IEEEtran}
\bibliography{Refs/CovRef}

\begin{thebibliography}{10}
\providecommand{\url}[1]{#1}
\csname url@samestyle\endcsname
\providecommand{\newblock}{\relax}
\providecommand{\bibinfo}[2]{#2}
\providecommand{\BIBentrySTDinterwordspacing}{\spaceskip=0pt\relax}
\providecommand{\BIBentryALTinterwordstretchfactor}{4}
\providecommand{\BIBentryALTinterwordspacing}{\spaceskip=\fontdimen2\font plus
\BIBentryALTinterwordstretchfactor\fontdimen3\font minus
  \fontdimen4\font\relax}
\providecommand{\BIBforeignlanguage}[2]{{%
\expandafter\ifx\csname l@#1\endcsname\relax
\typeout{** WARNING: IEEEtran.bst: No hyphenation pattern has been}%
\typeout{** loaded for the language `#1'. Using the pattern for}%
\typeout{** the default language instead.}%
\else
\language=\csname l@#1\endcsname
\fi
#2}}
\providecommand{\BIBdecl}{\relax}
\BIBdecl

\bibitem{cardenas2008secure}
A.~A. C\'{a}rdenas, S.~Amin, and S.~Sastry, ``Secure control: Towards
  survivable cyber-physical systems,'' in \emph{2008 The 28th International
  Conference on Distributed Computing Systems Workshops}.\hskip 1em plus 0.5em
  minus 0.4em\relax IEEE, 2008, pp. 495--500.

\bibitem{AAPC}
A.~A. C\'{a}rdenas, S.~Amin, Z.-S. Lin, Y.-L. Huang, C.-Y. Huang, and
  S.~Sastry, ``Attacks against process control systems: Risk assessment,
  detection, and response,'' in \emph{Proceedings of the 6th ACM Symposium on
  Information, Computer and Communications Security}, ser. ASIACCS '11.\hskip
  1em plus 0.5em minus 0.4em\relax New York, NY, USA: ACM, 2011, pp. 355--366.

\bibitem{rditp}
Y.~Mo, J.~Hespanha, and B.~Sinopoli, ``Robust detection in the presence of
  integrity attacks,'' in \emph{2012 American Control Conference (ACC)}, June
  2012, pp. 3541--3546.

\bibitem{adaii}
F.~Pasqualetti, F.~Dörfler, and F.~Bullo, ``Attack detection and
  identification in cyber-physical systems,'' \emph{IEEE Transactions on
  Automatic Control}, vol.~58, no.~11, pp. 2715--2729, Nov 2013.

\bibitem{ascf}
A.~Teixeira, I.~Shames, H.~Sandberg, and K.~H. Johansson, ``A secure control
  framework for resource-limited adversaries,'' \emph{Automatica}, vol.~51, no.
  Supplement C, pp. 135 -- 148, 2015.

\bibitem{10298792}
M.~Taheri, K.~Khorasani, I.~Shames, and N.~Meskin, ``Cyberattack and
  machine-induced fault detection and isolation methodologies for
  cyber-physical systems,'' \emph{IEEE Transactions on Control Systems
  Technology}, pp. 1--16, 2023.

\bibitem{sandberg2022secure}
H.~Sandberg, V.~Gupta, and K.~H. Johansson, ``Secure networked control
  systems,'' \emph{Annual Review of Control, Robotics, and Autonomous Systems},
  vol.~5, pp. 445--464, 2022.

\bibitem{9039550}
J.~Milošević, A.~Teixeira, K.~H. Johansson, and H.~Sandberg, ``Actuator
  security indices based on perfect undetectability: Computation, robustness,
  and sensor placement,'' \emph{IEEE Transactions on Automatic Control},
  vol.~65, no.~9, pp. 3816--3831, 2020.

\bibitem{smith2011decoupled}
R.~S. Smith, ``A decoupled feedback structure for covertly appropriating
  networked control systems,'' \emph{IFAC Proceedings Volumes}, vol.~44, no.~1,
  pp. 90--95, 2011.

\bibitem{9683376}
M.~Taheri, K.~Khorasani, I.~Shames, and N.~Meskin, ``Data-driven covert-attack
  strategies and countermeasures for cyber-physical systems,'' in \emph{2021
  60th IEEE Conference on Decision and Control (CDC)}, 2021, pp. 4170--4175.

\bibitem{kwon2013security}
C.~Kwon, W.~Liu, and I.~Hwang, ``Security analysis for cyber-physical systems
  against stealthy deception attacks,'' in \emph{2013 American control
  conference}.\hskip 1em plus 0.5em minus 0.4em\relax IEEE, 2013, pp.
  3344--3349.

\bibitem{9104026}
A.~Barboni, H.~Rezaee, F.~Boem, and T.~Parisini, ``Detection of covert
  cyber-attacks in interconnected systems: A distributed model-based
  approach,'' \emph{IEEE Transactions on Automatic Control}, vol.~65, no.~9,
  pp. 3728--3741, 2020.

\bibitem{docaazdaicps}
A.~Hoehn and P.~Zhang, ``Detection of covert attacks and zero dynamics attacks
  in cyber-physical systems,'' in \emph{2016 American Control Conference
  (ACC)}.\hskip 1em plus 0.5em minus 0.4em\relax IEEE, 2016, pp. 302--307.

\bibitem{ferrari2020switching}
R.~M. Ferrari and A.~M. Teixeira, ``A switching multiplicative watermarking
  scheme for detection of stealthy cyber-attacks,'' \emph{IEEE Transactions on
  Automatic Control}, vol.~66, no.~6, pp. 2558--2573, 2020.

\bibitem{7011170}
Y.~Mo, S.~Weerakkody, and B.~Sinopoli, ``Physical authentication of control
  systems: Designing watermarked control inputs to detect counterfeit sensor
  outputs,'' \emph{IEEE Control Systems Magazine}, vol.~35, no.~1, pp. 93--109,
  2015.

\bibitem{miao2016coding}
F.~Miao, Q.~Zhu, M.~Pajic, and G.~J. Pappas, ``Coding schemes for securing
  cyber-physical systems against stealthy data injection attacks,'' \emph{IEEE
  Transactions on Control of Network Systems}, vol.~4, no.~1, pp. 106--117,
  2016.

\bibitem{fang2019two}
S.~Fang, K.~H. Johansson, M.~Skoglund, H.~Sandberg, and H.~Ishii, ``Two-way
  coding in control systems under injection attacks: from attack detection to
  attack correction,'' in \emph{Proceedings of the 10th ACM/IEEE International
  Conference on Cyber-Physical Systems}, 2019, pp. 141--150.

\bibitem{9683075}
A.~J. Gallo, S.~C. Anand, A.~M.~H. Teixeira, and R.~M.~G. Ferrari, ``Design of
  multiplicative watermarking against covert attacks,'' in \emph{2021 60th IEEE
  Conference on Decision and Control (CDC)}, 2021, pp. 4176--4181.

\bibitem{ctfls}
H.~L. Trentelman, A.~A. Stoorvogel, and M.~Hautus, \emph{Control theory for
  linear systems}.\hskip 1em plus 0.5em minus 0.4em\relax Springer Science \&
  Business Media, 2012.

\bibitem{hajdasinski1979realization}
A.~Hajdasinski and A.~A.~H. Damen, \emph{Realization of the Markov parameter
  sequences using the singular value decomposition of the Hankel matrix}.\hskip
  1em plus 0.5em minus 0.4em\relax Technische Hogeschool Eindhoven, 1979.

\bibitem{de1988geometrical}
B.~De~Moor, J.~Vandewalle, M.~Moonen, L.~Vandenberghe, and P.~Van~Mieghem, ``A
  geometrical strategy for the identification of state space models of linear
  multivariable systems with singular value decomposition,'' \emph{IFAC
  Proceedings Volumes}, vol.~21, no.~9, pp. 493--497, 1988.

\bibitem{ljung1999system}
L.~Ljung, ``System identification,'' \emph{Wiley encyclopedia of electrical and
  electronics engineering}, pp. 1--19, 1999.

\bibitem{dong2011identification}
J.~Dong and M.~Verhaegen, ``Identification of fault estimation filter from i/o
  data for systems with stable inversion,'' \emph{IEEE Transactions on
  Automatic Control}, vol.~57, no.~6, pp. 1347--1361, 2011.

\bibitem{harkegaard2005resolving}
O.~H{\"a}rkeg{\aa}rd and S.~T. Glad, ``Resolving actuator redundancy—optimal
  control vs. control allocation,'' \emph{Automatica}, vol.~41, no.~1, pp.
  137--144, 2005.

\bibitem{boussaid2014ftc}
B.~Boussaid, C.~Aubrun, J.~Jiang, and M.~N. Abdelkrim, ``Ftc approach with
  actuator saturation avoidance based on reference management,''
  \emph{International Journal of Robust and Nonlinear Control}, vol.~24,
  no.~17, pp. 2724--2740, 2014.

\bibitem{safcpsuua}
Z.~Zhao, Y.~Yang, Y.~Li, and R.~Liu, ``Security analysis for cyber-physical
  systems under undetectable attacks: A geometric approach,''
  \emph{International Journal of Robust and Nonlinear Control}, vol.~30,
  no.~11, pp. 4359--4370, 2020.

\bibitem{siolcps}
A.~{Baniamerian} and K.~{Khorasani}, ``Security index of linear cyber-physical
  systems: A geometric perspective,'' in \emph{2019 6th International
  Conference on Control, Decision and Information Technologies (CoDIT)}, April
  2019, pp. 391--396.

\end{thebibliography}

\end{document}